\documentclass[12pt,reqno]{amsart}
\usepackage{graphicx}
\usepackage{amscd}
\usepackage{amsmath}
\usepackage{epsfig}
\usepackage{amsfonts}
\usepackage{amssymb}

\providecommand{\U}[1]{\protect\rule{.1in}{.1in}}
\providecommand{\U}[1]{\protect\rule{.1in}{.1in}}
\allowdisplaybreaks

\theoremstyle{plain}

\newtheorem{lemma}{Lemma}

\numberwithin{equation}{section}

\input{tcilatex}

\begin{document}
\title[Nonlinear Schr\"{o}dinger Equation]{An Integral Form of the Nonlinear
Schr\"{o}dinger Equation with Variable Coefficients}
\author{Erwin Suazo}
\address{Department of Mathematics and Statistics, Arizona State University,
Tempe, AZ 85287--1804, U.S.A.}
\email{suazo@mathpost.la.asu.edu}
\author{Sergei K. Suslov}
\address{Department of Mathematics and Statistics, Arizona State University,
Tempe, AZ 85287--1804, U.S.A.}
\email{sks@asu.edu}
\urladdr{http://hahn.la.asu.edu/\symbol{126}suslov/index.html}
\date{\today }

\begin{abstract}
We discuss an integral form of the Cauchy initial value problem for the
nonlinear Schr\"{o}dinger equation with variable coefficients. Some special
and limiting cases are outlined.
\end{abstract}

\maketitle

\section{Introduction}

In the previous Letter we have constructed the time evolution operator for
the linear one-dimensional time-dependent Schr\"{o}dinger equation%
\begin{equation}
i\hslash \frac{\partial \psi }{\partial t}=H\left( t\right) \psi
\label{int1}
\end{equation}%
in a general case when the Hamiltonian is an arbitrary quadratic form of the
operator of coordinate and the operator of linear momentum; see \cite%
{Cor-Sot:Lop:Sua:Sus}. In this approach, several exactly solvable models,
that have been studied elsewhere, are classified in terms of elementary
solutions of certain characteristic equation. A particular solution of the
corresponding nonlinear Schr\"{o}dinger equation with variable coefficients
had been obtained in a similar fashion. In this paper we rewrite a nonlinear
Schr\"{o}dinger equation%
\begin{equation}
\left( i\frac{\partial }{\partial t}-H\left( t\right) \right) \psi \left(
x,t\right) =F\left( t,x,\psi \left( x,t\right) \right)  \label{int2}
\end{equation}%
in an integral form and consider several examples. In general, we do not
assume that the time-dependent linear Hamiltonian $H\left( t\right) $ here
is the Hermitian operator.

\section{A General Lemma: Duhamel's Principle}

The following result is helpful in study of solutions of the nonlinear Schr%
\"{o}dinger equation (\ref{int2}) by a fixed point argument.

\begin{lemma}
Suppose that the Cauchy initial value problem%
\begin{equation}
\left( i\frac{\partial }{\partial t}-H\left( t\right) \right) \psi
_{0}\left( x,t\right) =0,\qquad \left. \psi _{0}\left( x,t\right)
\right\vert _{t=0}=\varphi \left( x\right)  \label{lem1}
\end{equation}%
for a linear time-dependent Schr\"{o}dinger equation can be solved in terms
of the time evolution operator%
\begin{equation}
\psi _{0}\left( x,t\right) =U\left( t\right) \psi _{0}\left( x,0\right)
\label{lem2}
\end{equation}%
with%
\begin{equation}
\left( i\frac{\partial }{\partial t}-H\left( t\right) \right) U\left(
t\right) =0  \label{lem3}
\end{equation}%
and%
\begin{equation}
U\left( 0\right) =I=U\left( t\right) U\left( t\right) ^{-1}.  \label{lem4}
\end{equation}%
Then the initial value problem%
\begin{equation}
\left( i\frac{\partial }{\partial t}-H\left( t\right) \right) \psi =F\left(
t,x,\psi \right) ,\qquad \left. \psi \left( x,t\right) \right\vert
_{t=0}=\chi \left( x\right)  \label{lem5}
\end{equation}%
for the nonlinear Schr\"{o}dinger equation can be rewritten as an integral
equation%
\begin{equation}
\psi \left( x,t\right) =U\left( t\right) \psi \left( x,0\right)
-i\int_{0}^{t}U\left( t\right) U\left( s\right) ^{-1}F\left( s,x,\psi \left(
x,s\right) \right) \ ds  \label{lem6}
\end{equation}%
in terms of the time evolution operator $U\left( t\right) $ for the linear
equation and its inverse.
\end{lemma}

\begin{proof}
Indeed, from (\ref{lem6})%
\begin{eqnarray*}
&&\left( i\frac{\partial }{\partial t}-H\left( t\right) \right) \psi =\left(
i\frac{\partial }{\partial t}-H\left( t\right) \right) U\left( t\right) \psi
\left( x,0\right) \\
&&\quad +\frac{\partial }{\partial t}\int_{0}^{t}U\left( t\right) U\left(
s\right) ^{-1}F\left( s,x,\psi \left( x,s\right) \right) \ ds \\
&&\quad \quad +i\int_{0}^{t}H\left( t\right) U\left( t\right) U\left(
s\right) ^{-1}F\left( s,x,\psi \left( x,s\right) \right) \ ds,
\end{eqnarray*}%
where%
\begin{eqnarray*}
&&\frac{\partial }{\partial t}\int_{0}^{t}U\left( t\right) U\left( s\right)
^{-1}F\left( s,x,\psi \left( x,s\right) \right) \ ds \\
&&\quad =U\left( t\right) U\left( t\right) ^{-1}F\left( t,x,\psi \left(
x,t\right) \right) \\
&&\qquad +\int_{0}^{t}\frac{\partial U\left( t\right) }{\partial t}U\left(
s\right) ^{-1}F\left( s,x,\psi \left( x,s\right) \right) \ ds.
\end{eqnarray*}%
Thus, by (\ref{lem3}) and (\ref{lem4})%
\begin{eqnarray*}
&&\left( i\frac{\partial }{\partial t}-H\left( t\right) \right) \psi
=F\left( t,x,\psi \left( x,t\right) \right) \\
&&\quad -i\int_{0}^{t}\left( i\frac{\partial }{\partial t}-H\left( t\right)
\right) U\left( t\right) \left( U\left( s\right) ^{-1}F\left( s,x,\psi
\left( x,s\right) \right) \right) \ ds \\
&&\qquad =F\left( t,x,\psi \left( x,t\right) \right) .
\end{eqnarray*}%
Initial conditions are satisfied in view of%
\begin{eqnarray*}
&&\lim_{t\rightarrow 0^{+}}\int_{0}^{t}U\left( t\right) U\left( s\right)
^{-1}F\left( s,x,\psi \left( x,s\right) \right) \ ds \\
&&\quad =U\left( 0\right) \lim_{t\rightarrow 0^{+}}\int_{0}^{t}U\left(
s\right) ^{-1}F\left( s,x,\psi \left( x,s\right) \right) \ ds=0,
\end{eqnarray*}%
where we have used the fact that if $\lim_{t\rightarrow 0^{+}}f\left(
t\right) =f\left( 0\right) ,$ then%
\begin{equation*}
\lim_{t\rightarrow 0^{+}}\int_{0}^{t}f\left( s\right) \ ds=0.
\end{equation*}%
This completes the proof.
\end{proof}

When $F$ does not depend on $\psi ,$ expression (\ref{lem6}) solves the
Cauchy initial value problem (\ref{lem5}) for the corresponding
nonhomogeneous Schr\"{o}dinger equation.

The integral equation (\ref{lem6}) can also be rewritten in the form%
\begin{eqnarray}
\psi \left( x,0\right) &=&U^{-1}\left( t\right) \psi \left( x,t\right)
+i\int_{0}^{t}U\left( s\right) ^{-1}F\left( s,x,\psi \left( x,s\right)
\right) \ ds  \label{lem7} \\
&=&T\left( t\right) ^{-1}\psi \left( x,t\right) ,  \notag
\end{eqnarray}%
which gives explicitly the inverse of the time evolution operator%
\begin{equation}
\psi \left( x,t\right) =T\left( t\right) \psi \left( x,0\right)  \label{lem8}
\end{equation}%
for the nonlinear Schr\"{o}dinger equation (\ref{int2}). Thus, in general,
solution of the Cauchy initial value problem (\ref{lem5}) is a problem of
inversion of this nonlinear integral operator.

\section{Quadratic Hamiltonians}

The fundamental solution of the linear Schr\"{o}dinger equation with the
quadratic Hamiltonian of the form%
\begin{equation}
i\frac{\partial \psi }{\partial t}=-a\left( t\right) \frac{\partial ^{2}\psi 
}{\partial x^{2}}+b\left( t\right) x^{2}\psi -i\left( c\left( t\right) x%
\frac{\partial \psi }{\partial x}+d\left( t\right) \psi \right) -f\left(
t\right) x\psi +ig\left( t\right) \frac{\partial \psi }{\partial x},
\label{schr1}
\end{equation}%
where $a\left( t\right) ,$ $b\left( t\right) ,$ $c\left( t\right) ,$ $%
d\left( t\right) ,$ $f\left( t\right) ,$ and $g\left( t\right) $ are given
real-valued functions of time $t$ only, can be found with the help of a
familiar substitution%
\begin{equation}
\psi =Ae^{iS}=A\left( t\right) e^{iS\left( x,y,t\right) },\qquad A=A\left(
t\right) =\frac{1}{\sqrt{2\pi i\mu \left( t\right) }}  \label{schr2}
\end{equation}%
with%
\begin{equation}
S=S\left( x,y,t\right) =\alpha \left( t\right) x^{2}+\beta \left( t\right)
xy+\gamma \left( t\right) y^{2}+\delta \left( t\right) x+\varepsilon \left(
t\right) y+\kappa \left( t\right) ,  \label{schr4}
\end{equation}%
where $\alpha \left( t\right) ,$ $\beta \left( t\right) ,$ $\gamma \left(
t\right) ,$ $\delta \left( t\right) ,$ $\varepsilon \left( t\right) ,$ and $%
\kappa \left( t\right) $ are differentiable real-valued functions of time $t$
only; see \cite{Cor-Sot:Lop:Sua:Sus}. Indeed,%
\begin{equation}
\frac{\partial S}{\partial t}=-a\left( \frac{\partial S}{\partial x}\right)
^{2}-bx^{2}+fx+\left( g-cx\right) \frac{\partial S}{\partial x}
\label{schr5}
\end{equation}%
provided%
\begin{equation}
\frac{\mu ^{\prime }}{2\mu }=a\frac{\partial ^{2}S}{\partial x^{2}}%
+d=2\alpha \left( t\right) a\left( t\right) +d\left( t\right) .
\label{schr6}
\end{equation}%
Equating the coefficients of all admissible powers of $x^{m}y^{n}$ with $%
0\leq m+n\leq 2,$ gives the following system of ordinary differential
equations%
\begin{align}
& \frac{d\alpha }{dt}+b\left( t\right) +2c\left( t\right) \alpha +4a\left(
t\right) \alpha ^{2}=0,  \label{schr7} \\
& \frac{d\beta }{dt}+\left( c\left( t\right) +4a\left( t\right) \alpha
\left( t\right) \right) \beta =0,  \label{schr8} \\
& \frac{d\gamma }{dt}+a\left( t\right) \beta ^{2}\left( t\right) =0,
\label{schr9} \\
& \frac{d\delta }{dt}+\left( c\left( t\right) +4a\left( t\right) \alpha
\left( t\right) \right) \delta =f\left( t\right) +2\alpha \left( t\right)
g\left( t\right) ,  \label{schr10} \\
& \frac{d\varepsilon }{dt}=\left( g\left( t\right) -2a\left( t\right) \delta
\left( t\right) \right) \beta \left( t\right) ,  \label{schr11} \\
& \frac{d\kappa }{dt}=g\left( t\right) \delta \left( t\right) -a\left(
t\right) \delta ^{2}\left( t\right) ,  \label{schr12}
\end{align}%
where the first equation is the familiar Riccati nonlinear differential
equation; see, for example, \cite{Haah:Stein}, \cite{Rainville}, \cite{Wa}
and references therein. Substitution of (\ref{schr6}) into (\ref{schr7})
results in the second order linear equation%
\begin{equation}
\mu ^{\prime \prime }-\tau \left( t\right) \mu ^{\prime }+4\sigma \left(
t\right) \mu =0  \label{schr13}
\end{equation}%
with%
\begin{equation}
\tau \left( t\right) =\frac{a^{\prime }}{a}-2c+4d,\qquad \sigma \left(
t\right) =ab-cd+d^{2}+\frac{d}{2}\left( \frac{a^{\prime }}{a}-\frac{%
d^{\prime }}{d}\right) ,  \label{schr13a}
\end{equation}%
which must be solved subject to the initial data%
\begin{equation}
\mu \left( 0\right) =0,\qquad \mu ^{\prime }\left( 0\right) =2a\left(
0\right) \neq 0  \label{schr13b}
\end{equation}%
in order to satisfy the initial condition for the corresponding Green
function; see the asymptotic formula (\ref{schr21}) below. We refer to
equation (\ref{schr13}) as the \textit{characteristic equation} and its
solution $\mu \left( t\right) ,$ subject to (\ref{schr13b}), as the \textit{%
characteristic function.} As the special case (\ref{schr13}) contains the
generalized equation of hypergeometric type, whose solutions are studied in
detail in \cite{Ni:Uv}; see also \cite{An:As:Ro}, \cite{Ni:Su:Uv}, \cite%
{Sus:Trey}, and \cite{Wa}.

Thus, the Green function (fundamental solution or propagator) is given in
terms of the characteristic function%
\begin{equation}
\psi =G\left( x,y,t\right) =\frac{1}{\sqrt{2\pi i\mu \left( t\right) }}\
e^{i\left( \alpha \left( t\right) x^{2}+\beta \left( t\right) xy+\gamma
\left( t\right) y^{2}+\delta \left( t\right) x+\varepsilon \left( t\right)
y+\kappa \left( t\right) \right) }.  \label{schr14}
\end{equation}%
Here%
\begin{equation}
\alpha \left( t\right) =\frac{1}{4a\left( t\right) }\frac{\mu ^{\prime
}\left( t\right) }{\mu \left( t\right) }-\frac{d\left( t\right) }{2a\left(
t\right) },  \label{schr15}
\end{equation}%
\begin{equation}
\beta \left( t\right) =-\frac{1}{\mu \left( t\right) }\ \exp \left(
-\int_{0}^{t}\left( c\left( \tau \right) -2d\left( \tau \right) \right) \
d\tau \right) ,  \label{scr16}
\end{equation}%
\begin{eqnarray}
\gamma \left( t\right) &=&\frac{a\left( t\right) }{\mu \left( t\right) \mu
^{\prime }\left( t\right) }\ \exp \left( -2\int_{0}^{t}\left( c\left( \tau
\right) -2d\left( \tau \right) \right) \ d\tau \right)  \label{schr17} \\
&&\quad -4\int_{0}^{t}\frac{a\left( \tau \right) \sigma \left( \tau \right) 
}{\left( \mu ^{\prime }\left( \tau \right) \right) ^{2}}\left( \exp \left(
-2\int_{0}^{\tau }\left( c\left( \lambda \right) -2d\left( \lambda \right)
\right) \ d\lambda \right) \right) \ d\tau ,  \notag
\end{eqnarray}%
\begin{eqnarray}
\delta \left( t\right) &=&\frac{1}{\mu \left( t\right) }\ \exp \left(
-\int_{0}^{t}\left( c\left( \tau \right) -2d\left( \tau \right) \right) \
d\tau \right) \   \label{schr18} \\
&&\quad \times \int_{0}^{t}\exp \left( \int_{0}^{\tau }\left( c\left(
\lambda \right) -2d\left( \lambda \right) \right) \ d\lambda \right)  \notag
\\
&&\quad \quad \times \left( \left( f\left( \tau \right) -\frac{d\left( \tau
\right) }{a\left( \tau \right) }g\left( \tau \right) \right) \mu \left( \tau
\right) +\frac{g\left( \tau \right) }{2a\left( \tau \right) }\mu ^{\prime
}\left( \tau \right) \right) \ d\tau ,  \notag
\end{eqnarray}%
\begin{eqnarray}
\varepsilon \left( t\right) &=&-\frac{2a\left( t\right) }{\mu ^{\prime
}\left( t\right) }\delta \left( t\right) \ \exp \left( -\int_{0}^{t}\left(
c\left( \tau \right) -2d\left( \tau \right) \right) \ d\tau \right)
\label{schr18a} \\
&&\quad +8\int_{0}^{t}\frac{a\left( \tau \right) \sigma \left( \tau \right) 
}{\left( \mu ^{\prime }\left( \tau \right) \right) ^{2}}\exp \left(
-\int_{0}^{\tau }\left( c\left( \lambda \right) -2d\left( \lambda \right)
\right) \ d\lambda \right) \left( \mu \left( \tau \right) \delta \left( \tau
\right) \right) \ d\tau  \notag \\
&&\quad +2\int_{0}^{t}\frac{a\left( \tau \right) }{\mu ^{\prime }\left( \tau
\right) }\exp \left( -\int_{0}^{\tau }\left( c\left( \lambda \right)
-2d\left( \lambda \right) \right) \ d\lambda \right) \left( f\left( \tau
\right) -\frac{d\left( \tau \right) }{a\left( \tau \right) }g\left( \tau
\right) \right) \ d\tau ,  \notag
\end{eqnarray}%
\begin{eqnarray}
\kappa \left( t\right) &=&\frac{a\left( t\right) \mu \left( t\right) }{\mu
^{\prime }\left( t\right) }\delta ^{2}\left( t\right) -4\int_{0}^{t}\frac{%
a\left( \tau \right) \sigma \left( \tau \right) }{\left( \mu ^{\prime
}\left( \tau \right) \right) ^{2}}\left( \mu \left( \tau \right) \delta
\left( \tau \right) \right) ^{2}\ d\tau  \label{schr19} \\
&&\quad -2\int_{0}^{t}\frac{a\left( \tau \right) }{\mu ^{\prime }\left( \tau
\right) }\left( \mu \left( \tau \right) \delta \left( \tau \right) \right)
\left( f\left( \tau \right) -\frac{d\left( \tau \right) }{a\left( \tau
\right) }g\left( \tau \right) \right) \ d\tau  \notag
\end{eqnarray}%
with $\delta \left( 0\right) =g\left( 0\right) /\left( 2a\left( 0\right)
\right) ,$ $\varepsilon \left( 0\right) =-\delta \left( 0\right) ,$ and $%
\kappa \left( 0\right) =0.$ Integration by parts has been used to resolve
the singularities of the initial data. Then the corresponding asymptotic
formula is%
\begin{equation}
G\left( x,y,t\right) =\frac{e^{iS\left( x,y,t\right) }}{\sqrt{2\pi i\mu
\left( t\right) }}\sim \frac{1}{\sqrt{4\pi ia\left( 0\right) t}}\exp \left( i%
\frac{\left( x-y\right) ^{2}}{4a\left( 0\right) t}\right) \exp \left( i\frac{%
g\left( 0\right) }{2a\left( 0\right) }\left( x-y\right) \right)
\label{schr21}
\end{equation}%
as $t\rightarrow 0^{+}.$ Notice that the first term on the right hand side
is a familiar free particle propagator (cf.~(\ref{sp1}) below).

By the superposition principle, an explicit solution of the Cauchy initial
value problem%
\begin{equation}
i\frac{\partial \psi }{\partial t}=H\left( t\right) \psi ,\qquad \left. \psi
\left( x,t\right) \right\vert _{t=0}=\varphi \left( x\right)  \label{schr22}
\end{equation}%
on the infinite interval $-\infty <x<\infty $ with the general quadratic
Hamiltonian as in (\ref{schr1}) has the form%
\begin{equation}
\psi \left( x,t\right) =\int_{-\infty }^{\infty }G\left( x,y,t\right) \ \psi
\left( y,0\right) \ dy.  \label{schr23}
\end{equation}%
This yields the time evolution operator (\ref{lem2}) explicitly as an
integral operator.

\section{Inverses of the Time Evolution Operator}

In the previous section we have discussed how to construct the time
evolution operator for the linear Schr\"{o}dinger equation with the
quadratic Hamiltonian (\ref{schr1}). In this section we study the inverses%
\begin{eqnarray}
\psi \left( x,t\right) &=&U\left( t\right) \psi \left( x,0\right)
=\int_{-\infty }^{\infty }G\left( x,y,t\right) \ \psi \left( y,0\right) \ dy,
\label{inv1} \\
\psi \left( x,0\right) &=&U^{-1}\left( t\right) \psi \left( x,t\right)
=\int_{-\infty }^{\infty }H\left( x,y,t\right) \ \psi \left( y,t\right) \ dy
\label{inv2}
\end{eqnarray}%
such that%
\begin{equation}
U\left( t\right) U^{-1}\left( t\right) =U^{-1}\left( t\right) U\left(
t\right) =I=\text{id}.  \label{inv3}
\end{equation}%
Here we introduce%
\begin{eqnarray}
G\left( x,y,t\right) &=&\frac{e^{iS\left( x,y,t\right) }}{\sqrt{2\pi i\mu
\left( t\right) }},  \label{inv4} \\
H\left( x,y,t\right) &=&G^{\ast }\left( y,x,t\right) \exp \left(
-\int_{0}^{t}\left( c\left( \tau \right) -2d\left( \tau \right) \right) \
d\tau \right)  \label{inv5} \\
&=&\frac{e^{-iS\left( y,x,t\right) }}{\sqrt{-2\pi i\mu \left( t\right) }}%
\exp \left( -\int_{0}^{t}\left( c\left( \tau \right) -2d\left( \tau \right)
\right) \ d\tau \right)  \notag
\end{eqnarray}%
with $S\left( x,y,t\right) =\alpha \left( t\right) x^{2}+\beta \left(
t\right) xy+\gamma \left( t\right) y^{2}+\delta \left( t\right)
x+\varepsilon \left( t\right) y+\kappa \left( t\right) .$

First we formally prove the following orthogonality relations of the kernels%
\begin{eqnarray}
\int_{-\infty }^{\infty }G\left( x,y,t\right) H\left( y,z,t\right) \ dy
&=&e^{i\left( \alpha \left( t\right) \left( x+z\right) +\delta \left(
t\right) \right) \left( x-z\right) }\ \delta \left( x-z\right) ,
\label{inv6} \\
\int_{-\infty }^{\infty }H\left( x,y,t\right) G\left( y,z,t\right) \ dy
&=&e^{-i\left( \gamma \left( t\right) \left( x+z\right) +\varepsilon \left(
t\right) \right) \left( x-z\right) }\ \delta \left( x-z\right) ,
\label{inv7}
\end{eqnarray}%
where $\delta \left( x\right) $ is the Dirac delta function with respect to
the space coordinates (do not confuse with a given function of time $\delta
\left( t\right) $ throughout the paper).

Indeed, by (\ref{inv4})--(\ref{inv5}) one gets%
\begin{eqnarray*}
\int_{-\infty }^{\infty }G\left( x,y,t\right) H\left( y,z,t\right) \ dy &=&%
\frac{1}{\mu \left( t\right) }\exp \left( -\int_{0}^{t}\left( c\left( \tau
\right) -2d\left( \tau \right) \right) \ d\tau \right) \\
&&\times e^{i\left( \alpha \left( t\right) \left( x+z\right) +\delta \left(
t\right) \right) \left( x-z\right) }\ \frac{1}{2\pi }\int_{-\infty }^{\infty
}e^{i\beta \left( t\right) \left( x-z\right) y}\ dy \\
&=&e^{i\left( \alpha \left( t\right) \left( x+z\right) +\delta \left(
t\right) \right) \left( x-z\right) }\ \delta \left( x-z\right)
\end{eqnarray*}%
in view of (\ref{scr16}) with $-\beta \left( t\right) >0$ and the integral%
\begin{equation}
\delta \left( \zeta \right) =\frac{1}{2\pi }\int_{-\infty }^{\infty
}e^{i\zeta \xi }\ d\xi  \label{dirac}
\end{equation}%
as the Dirac delta function. The formal proof of (\ref{inv7}) is similar and
is left to the reader.

Now we have%
\begin{eqnarray*}
&&U^{-1}\left( t\right) U\left( t\right) \psi \left( x,0\right)
=U^{-1}\left( t\right) \psi \left( x,t\right) \\
&&\quad =\int_{-\infty }^{\infty }H\left( x,y,t\right) \ \psi \left(
y,t\right) \ dy \\
&&\quad =\int_{-\infty }^{\infty }H\left( x,y,t\right) \ \left(
\int_{-\infty }^{\infty }G\left( y,z,t\right) \ \psi \left( z,0\right) \
dz\right) \ dy \\
&&\quad =\int_{-\infty }^{\infty }\left( \int_{-\infty }^{\infty }H\left(
x,y,t\right) G\left( y,z,t\right) \ dy\right) \ \psi \left( z,0\right) \ dz
\\
&&\quad =\int_{-\infty }^{\infty }e^{-i\left( \gamma \left( t\right) \left(
x+z\right) +\varepsilon \left( t\right) \right) \left( x-z\right) }\ \delta
\left( x-z\right) \ \psi \left( z,0\right) \ dz \\
&&\quad =\psi \left( x,0\right) ,
\end{eqnarray*}%
or $U^{-1}\left( t\right) U\left( t\right) =I.$ A formal proof of the second
relation $U\left( t\right) U^{-1}\left( t\right) =I$ is similar and left to
the reader.

An integral form (\ref{lem6}) of the nonlinear Schr\"{o}dinger equation (\ref%
{int2}) contains also the integral operator $U\left( t,s\right) =U\left(
t\right) U^{-1}\left( s\right) :$%
\begin{equation}
U\left( t\right) U^{-1}\left( s\right) \psi \left( x,s\right) =\int_{-\infty
}^{\infty }G\left( x,y,t,s\right) \psi \left( y,s\right) \ dy  \label{inv8a}
\end{equation}%
with the kernel given by%
\begin{equation}
G\left( x,y,t,s\right) =\int_{-\infty }^{\infty }G\left( x,z,t\right)
H\left( z,y,s\right) \ dz.  \label{inv8}
\end{equation}%
Here%
\begin{eqnarray}
&&\int_{-\infty }^{\infty }G\left( x,z,t\right) H\left( z,y,s\right) \ dz=%
\frac{1}{2\pi \sqrt{\mu \left( t\right) \mu \left( s\right) }}  \label{inv9}
\\
&&\qquad \times \exp \left( -\int_{0}^{s}\left( c\left( \tau \right)
-2d\left( \tau \right) \right) \ d\tau \right)  \notag \\
&&\qquad \times e^{i\left( \alpha \left( t\right) x^{2}-\alpha \left(
s\right) y^{2}+\delta \left( t\right) x-\delta \left( s\right) y+\kappa
\left( t\right) -\kappa \left( s\right) \right) }  \notag \\
&&\qquad \times \int_{-\infty }^{\infty }e^{i\left( \left( \gamma \left(
t\right) -\gamma \left( s\right) \right) z^{2}+\left( \beta \left( t\right)
x-\beta \left( s\right) y+\varepsilon \left( t\right) -\varepsilon \left(
s\right) \right) z\right) }\ dz  \notag
\end{eqnarray}%
and with the help of the familiar elementary integral%
\begin{equation}
\int_{-\infty }^{\infty }e^{i\left( az^{2}+2bz\right) }\,dz=\sqrt{\frac{\pi i%
}{a}}\,e^{-ib^{2}/a},  \label{gauss}
\end{equation}%
see Refs.~\cite{Bo:Shi} and \cite{Palio:Mead}, we get%
\begin{eqnarray}
&&G\left( x,y,t,s\right) =\frac{1}{\sqrt{4\pi i\mu \left( t\right) \mu
\left( s\right) \left( \gamma \left( s\right) -\gamma \left( t\right)
\right) }}  \label{inv9a} \\
&&\qquad \times \exp \left( -\int_{0}^{s}\left( c\left( \tau \right)
-2d\left( \tau \right) \right) \ d\tau \right)  \notag \\
&&\qquad \times \exp \left( i\left( \alpha \left( t\right) x^{2}-\alpha
\left( s\right) y^{2}+\delta \left( t\right) x-\delta \left( s\right)
y+\kappa \left( t\right) -\kappa \left( s\right) \right) \right)  \notag \\
&&\qquad \times \exp \left( \frac{\left( \beta \left( t\right) x-\beta
\left( s\right) y+\varepsilon \left( t\right) -\varepsilon \left( s\right)
\right) ^{2}}{4i\left( \gamma \left( t\right) -\gamma \left( s\right)
\right) }\right) .  \notag
\end{eqnarray}%
This can be transform into a somewhat more convenient form%
\begin{eqnarray}
&&G\left( x,y,t,s\right) =\frac{1}{\sqrt{4\pi i\mu \left( t\right) \mu
\left( s\right) \left( \gamma \left( s\right) -\gamma \left( t\right)
\right) }}\exp \left( -\int_{0}^{s}\left( c\left( \tau \right) -2d\left(
\tau \right) \right) \ d\tau \right)  \label{inv9b} \\
&&\quad \times \exp \left( \frac{\left( \varepsilon \left( t\right)
-\varepsilon \left( s\right) \right) ^{2}-4\left( \gamma \left( t\right)
-\gamma \left( s\right) \right) \left( \kappa \left( t\right) -\kappa \left(
s\right) \right) }{4i\left( \gamma \left( t\right) -\gamma \left( s\right)
\right) }\right)  \notag \\
&&\quad \times \exp \left( \frac{\left( \varepsilon \left( t\right)
-\varepsilon \left( s\right) \right) \left( \beta \left( t\right) x-\beta
\left( s\right) y\right) -2\left( \gamma \left( t\right) -\gamma \left(
s\right) \right) \left( \delta \left( t\right) x-\delta \left( s\right)
y\right) }{2i\left( \gamma \left( t\right) -\gamma \left( s\right) \right) }%
\right)  \notag \\
&&\quad \times \exp \left( \frac{\left( \beta \left( t\right) x-\beta \left(
s\right) y\right) ^{2}-4\left( \gamma \left( t\right) -\gamma \left(
s\right) \right) \left( \alpha \left( t\right) x^{2}-\alpha \left( s\right)
y^{2}\right) }{4i\left( \gamma \left( t\right) -\gamma \left( s\right)
\right) }\right) .  \notag
\end{eqnarray}%
In the limit $s\rightarrow t$ with $s<t$ one arrives at the kernel of the
identity operator. We leave the details to the reader.

\section{Axillary Tools: An Estimate and a Functional Equation}

Consider the operator $U\left( t,s\right) =U\left( t\right) U^{-1}\left(
s\right) .$ From (\ref{inv8a}) and (\ref{inv9b}) one gets%
\begin{eqnarray*}
\left\vert U\left( t,s\right) \psi \left( x,s\right) \right\vert
&=&\left\vert \int_{-\infty }^{\infty }G\left( x,y,t,s\right) \psi \left(
y,s\right) \ dy\right\vert \\
&\leq &\int_{-\infty }^{\infty }\left\vert G\left( x,y,t,s\right) \psi
\left( y,s\right) \right\vert \ dy \\
&=&\frac{1}{\sqrt{4\pi \left\vert \mu \left( t\right) \mu \left( s\right)
\left( \gamma \left( s\right) -\gamma \left( t\right) \right) \right\vert }}
\\
&&\times \exp \left( -\int_{0}^{s}\left( c\left( \tau \right) -2d\left( \tau
\right) \right) \ d\tau \right) \\
&&\times \int_{-\infty }^{\infty }\left\vert \psi \left( y,s\right)
\right\vert \ dy
\end{eqnarray*}%
and as a result%
\begin{eqnarray}
\left\vert U\left( t,s\right) \psi \left( x,s\right) \right\vert &\leq &%
\frac{1}{\sqrt{4\pi \left\vert \mu \left( t\right) \mu \left( s\right)
\left( \gamma \left( s\right) -\gamma \left( t\right) \right) \right\vert }}
\label{at1} \\
&&\times \exp \left( -\int_{0}^{s}\left( c\left( \tau \right) -2d\left( \tau
\right) \right) \ d\tau \right)  \notag \\
&&\quad \times \int_{-\infty }^{\infty }\left\vert \psi \left( y,s\right)
\right\vert \ dy.  \notag
\end{eqnarray}%
Thus the familiar estimate%
\begin{eqnarray}
\left\Vert U\left( t,s\right) \psi \right\Vert _{\infty } &\leq &\frac{1}{%
\sqrt{4\pi \left\vert \mu \left( t\right) \mu \left( s\right) \left( \gamma
\left( s\right) -\gamma \left( t\right) \right) \right\vert }}  \label{at2}
\\
&&\times \exp \left( -\int_{0}^{s}\left( c\left( \tau \right) -2d\left( \tau
\right) \right) \ d\tau \right) \ \left\Vert \psi \right\Vert _{1}  \notag
\end{eqnarray}%
holds in the case of the general quadratic Hamiltonian (\ref{schr1}) (cf.~%
\cite{Keel:Tao}).

As we shall see in the next section, some solutions of the characteristic
equation (\ref{schr13}) obey the following property%
\begin{equation}
\mu \left( t\right) \mu \left( s\right) \left( \gamma \left( s\right)
-\gamma \left( t\right) \right) =\chi \left( \frac{t+s}{2}\right) \mu \left(
t-s\right) .  \label{at3}
\end{equation}%
Then%
\begin{equation*}
\mu \left( t\right) \mu \left( s\right) \frac{\gamma \left( s\right) -\gamma
\left( t\right) }{t-s}=\chi \left( \frac{t+s}{2}\right) \frac{\mu \left(
t-s\right) -\mu \left( 0\right) }{t-s}
\end{equation*}%
and in the limit $s\rightarrow t$ one gets%
\begin{equation*}
-\mu ^{2}\left( t\right) \frac{d\gamma }{dt}=\chi \left( t\right) \mu
^{\prime }\left( 0\right) ,
\end{equation*}%
or by (\ref{schr13b})%
\begin{equation}
\frac{d\gamma }{dt}+2a\left( 0\right) \frac{\chi \left( t\right) }{\mu
^{2}\left( t\right) }=0.  \label{at4}
\end{equation}%
But, in view of (\ref{schr9}) and (\ref{scr16}),%
\begin{equation}
\frac{d\gamma }{dt}+\exp \left( -2\dint_{0}^{t}\left( c\left( \tau \right)
-2d\left( \tau \right) \right) \ d\tau \right) \ \frac{a\left( t\right) }{%
\mu ^{2}\left( t\right) }=0.  \label{at5}
\end{equation}%
Therefore, the addition property (\ref{at3}) may hold only when%
\begin{equation}
\chi \left( t\right) =\frac{1}{2}\exp \left( -2\dint_{0}^{t}\left( c\left(
\tau \right) -2d\left( \tau \right) \right) \ d\tau \right) \ \frac{a\left(
t\right) }{a\left( 0\right) }.  \label{at6}
\end{equation}%
Some examples will be given in the next section.

\section{Examples}

Now let us consider several elementary solutions of the characteristic
equation (\ref{schr13}); more complicated cases may include special
functions, like Bessel, hypergeometric or elliptic functions \cite{An:As:Ro}%
, \cite{Ni:Uv}, \cite{Rain}, and \cite{Wa}. Among special cases of general
expressions for the Green function (\ref{schr14})--(\ref{schr19}) are the
following \cite{Cor-Sot:Lop:Sua:Sus}:

\subsection{A Free Particle}

When $a=1/2,$ $b=c=d=f=g=0,$ and $\mu ^{\prime \prime }=0,$ $\mu =t,$ one
gets%
\begin{equation}
G\left( x,y,t\right) =\frac{1}{\sqrt{2\pi it}}\ \exp \left( \frac{i\left(
x-y\right) ^{2}}{2t}\right)  \label{sp1}
\end{equation}%
as the free particle propagator \cite{Fey:Hib}. In this case $\alpha =-\beta
/2=\gamma =1/\left( 2t\right) $ and an elementary identity%
\begin{equation*}
\frac{\left( x/t-y/s\right) ^{2}}{\left( 1/t-1/s\right) }-\frac{x^{2}}{t}+%
\frac{y^{2}}{s}=-\frac{\left( x-y\right) ^{2}}{t-s}
\end{equation*}%
implies that%
\begin{equation}
G\left( x,y,t,s\right) =G\left( x,y,t-s\right)  \label{sp1a}
\end{equation}%
from the general formula (\ref{inv9b}). The time evolution operator of the
linear problem is given explicitly as the following integral operator%
\begin{equation}
U\left( t\right) \chi \left( x\right) =\frac{1}{\sqrt{2\pi it}}\
\int_{-\infty }^{\infty }e^{i\left( x-y\right) ^{2}/2t}\ \chi \left(
y\right) \ dy  \label{sp1b}
\end{equation}%
and the traditional nonlinear Schr\"{o}dinger equation 
\begin{equation}
\left( i\frac{\partial }{\partial t}+\frac{1}{2}\frac{\partial ^{2}}{%
\partial x^{2}}\right) \psi =\lambda \left\vert \psi \right\vert ^{2\nu
}\psi ,\qquad \lambda =\text{constant},\quad 0<\nu \leq 1  \label{sp1c}
\end{equation}%
has the familiar integral form%
\begin{equation}
\psi \left( x,t\right) =U\left( t\right) \psi \left( x,0\right) -i\lambda
\int_{0}^{t}U\left( t-s\right) \left\vert \psi \left( x,s\right) \right\vert
^{2\nu }\psi \left( x,s\right) \ ds.  \label{sp1d}
\end{equation}%
See \cite{Cazenave}, \cite{Fadd:Takh}, \cite{Keel:Tao}, \cite{Tao} and
references therein for more information.

\subsection{A Particle in a Uniform External Field}

For a particle in a constant external field, where $a=1/2,$ $b=c=d=g=0$ and $%
f=\ $constant, $\mu =t,$ the propagator of the linear problem is given by 
\begin{equation}
G\left( x,y,t\right) =\frac{1}{\sqrt{2\pi it}}\ \exp \left( \frac{i\left(
x-y\right) ^{2}}{2t}\right) \exp \left( \frac{if\left( x+y\right) }{2}t-%
\frac{if^{2}}{24}t^{3}\right) .  \label{sp2}
\end{equation}%
This case was studied in detail in \cite{Arrighini:Durante}, \cite%
{Brown:Zhang}, \cite{Fey:Hib}, \cite{Holstein97}, \cite{Nardone} and \cite%
{Robinett}. We have corrected a typo in \cite{Fey:Hib}; see \cite{Styer} for
a complete list of known errata in the Feynman and Hibbs book.

In this case once again $\alpha =-\beta /2=\gamma =1/\left( 2t\right) $ and,
in addition to the case of a free particle, $\delta =\varepsilon =\left(
ft\right) /2$ and $\kappa =-f^{2}t^{3}/24.$ An elementary calculation shows
from the general expression (\ref{inv9b}) that relation (\ref{sp1a}) holds.
Therefore, the corresponding nonlinear Schr\"{o}dinger equation 
\begin{equation}
\left( i\frac{\partial }{\partial t}+\frac{1}{2}\frac{\partial ^{2}}{%
\partial x^{2}}+fx\right) \psi =\lambda \left\vert \psi \right\vert ^{2\nu
}\psi ,\qquad \lambda =\text{constant},\quad 0<\nu \leq 1  \label{sp2a}
\end{equation}%
has the integral form (\ref{sp1d}), where the linear propagator is given by (%
\ref{sp2}).

In a more general case of a particle in a uniform electric field changing in
time with a similar velocity-dependent term%
\begin{equation}
\left( i\frac{\partial }{\partial t}+\frac{1}{2}\frac{\partial ^{2}}{%
\partial x^{2}}+f\left( t\right) x-ig\left( t\right) \frac{\partial }{%
\partial x}\right) \psi =\lambda \left\vert \psi \right\vert ^{2\nu }\psi ,
\label{sp2b}
\end{equation}%
where $f\left( t\right) $ and $g\left( t\right) $ are functions of time
only, the propagator of the linear problem has the form%
\begin{equation}
G\left( x,y,t\right) =\frac{1}{\sqrt{2\pi it}}\ \exp \left( \frac{i\left(
x-y\right) ^{2}}{2t}\right) \exp \left( i\left( \delta \left( t\right)
x+\varepsilon \left( t\right) y+\kappa \left( t\right) \right) \right)
\label{sp2c}
\end{equation}%
with%
\begin{eqnarray}
\delta \left( t\right) &=&\frac{1}{t}\int_{0}^{t}\left( f\left( \tau \right)
\tau +g\left( \tau \right) \right) \ d\tau ,  \label{sp2d} \\
\varepsilon \left( t\right) &=&-\delta \left( t\right) +\int_{0}^{t}f\left(
\tau \right) \ d\tau  \label{sp2e}
\end{eqnarray}%
and%
\begin{equation}
\kappa \left( t\right) =\frac{t}{2}\delta ^{2}\left( t\right)
-\int_{0}^{t}\tau \delta \left( t\right) f\left( \tau \right) \ d\tau .
\label{sp2f}
\end{equation}%
A semigroup property \cite{Cazenave}, related to (\ref{sp1a}), does not hold
anymore and one has to use a general expression (\ref{inv9b}) in order to
write the integral equation (\ref{lem6}). But, in view of an elementary
identity%
\begin{equation}
\mu \left( t\right) \mu \left( s\right) \left( \gamma \left( s\right)
-\gamma \left( t\right) \right) =\frac{1}{2}\mu \left( t-s\right) ,
\label{sp2g}
\end{equation}%
an important addition formula still holds for the amplitude of the kernel $%
G\left( x,y,t,s\right) $ in (\ref{inv8a})--(\ref{inv8}) and (\ref{inv9b}) of
the operator $U\left( t,s\right) =U\left( t\right) U^{-1}\left( s\right) .$

\subsection{The Forced Harmonic Oscillator}

The simple harmonic oscillator with $a=b=1/2,$ $c=d=f=g=0$ and $\mu ^{\prime
\prime }+\mu =0,$ $\mu =\sin t$ has the familiar propagator of the form 
\begin{equation}
G\left( x,y,t\right) =\frac{1}{\sqrt{2\pi i\sin t}}\exp \left( \dfrac{i}{%
2\sin t}\left( \left( x^{2}+y^{2}\right) \cos t-2xy\right) \right) ,
\label{sp3}
\end{equation}%
which is studied in detail at \cite{Beauregard}, \cite{Gottf:T-MY}, \cite%
{Holstein}, \cite{Maslov:Fedoriuk}, \cite{Merz}, \cite{Thomber:Taylor}.

Once again relation (\ref{sp1a}) holds and the corresponding nonlinear Schr%
\"{o}dinger equation 
\begin{equation}
i\frac{\partial \psi }{\partial t}+\frac{1}{2}\left( \frac{\partial ^{2}}{%
\partial x^{2}}-x^{2}\right) \psi =\lambda \left\vert \psi \right\vert
^{2\nu }\psi ,\qquad \lambda =\text{constant},\quad 0<\nu \leq 1
\label{sp3a}
\end{equation}%
has the integral form (\ref{sp1d}), where the propagator is given by (\ref%
{sp3}). See \cite{Cazenave}, \cite{Car}, \cite{Carl}, \cite{Carle}, \cite%
{Carles}, \cite{Carles:Miller}, \cite{Carles:Nakamura} and references
therein for investigation of solutions of this integral equation by a fixed
point argument.

A linear problem extension to the case of the forced harmonic oscillator
including an extra velocity-dependent term and a time-dependent frequency is
discussed in \cite{FeynmanPhD}, \cite{Feynman}, \cite{Fey:Hib} and \cite%
{Lop:Sus}. The nonlinear Schr\"{o}dinger equation of interest is%
\begin{equation}
i\frac{\partial \psi }{\partial t}+\frac{1}{2}\left( \frac{\partial ^{2}}{%
\partial x^{2}}-x^{2}\right) \psi +f\left( t\right) x\psi -ig\left( t\right) 
\frac{\partial \psi }{\partial x}=\lambda \left\vert \psi \right\vert ^{2\nu
}\psi  \label{sp4}
\end{equation}%
and the corresponding propagator has the form%
\begin{eqnarray}
G\left( x,y,t\right) &=&\frac{1}{\sqrt{2\pi i\sin t}}\exp \left( \dfrac{i}{%
2\sin t}\left( \left( x^{2}+y^{2}\right) \cos t-2xy\right) \right)
\label{sp6} \\
&&\times \exp \left( i\left( \delta \left( t\right) x+\varepsilon \left(
t\right) y+\kappa \left( t\right) \right) \right)  \notag
\end{eqnarray}%
with%
\begin{eqnarray}
\delta \left( t\right) &=&\frac{1}{\sin t}\int_{0}^{t}\left( f\left( \tau
\right) \sin \tau +g\left( \tau \right) \cos \tau \right) \ d\tau ,
\label{sp6a} \\
\varepsilon \left( t\right) &=&-\frac{\delta \left( t\right) }{\cos t}%
+\int_{0}^{t}\frac{\sin \tau \ \delta \left( \tau \right) }{\cos ^{2}\tau }\
d\tau +\int_{0}^{t}\frac{f\left( \tau \right) }{\cos \tau }\ d\tau
\label{sp6b}
\end{eqnarray}%
and%
\begin{eqnarray}
\kappa \left( t\right) &=&\frac{1}{2}\tan t\ \delta ^{2}\left( t\right) -%
\frac{1}{2}\int_{0}^{t}\tan ^{2}\tau \ \delta ^{2}\left( \tau \right) \ d\tau
\label{sp6c} \\
&&-\int_{0}^{t}\tan \tau \ \delta \left( \tau \right) f\left( \tau \right) \
d\tau .  \notag
\end{eqnarray}%
The addition property (\ref{at3}) holds in this case. We leave the detail to
the reader.

\subsection{A Modified Oscillator}

Furthermore, an exact solution of the $n$-dimensional time-dependent Schr%
\"{o}dinger equation for certain modified oscillator is found in \cite%
{Me:Co:Su}. In the one-dimensional case we get functions $a=\cos ^{2}t,$ $%
b=\sin ^{2}t,$ $c=2d=\sin 2t$ and our characteristic equation (\ref{schr13})
takes the form%
\begin{equation}
\mu ^{\prime \prime }+2\tan t\ \mu ^{\prime }-2\mu =0,  \label{sp5}
\end{equation}%
whose elementary solution $\mu =\cos t\sinh t+\sin t\cosh t$ satisfies the
initial conditions (\ref{schr13b}). Further, the corresponding propagator is
given by%
\begin{align}
G\left( x,y,t\right) & =\frac{1}{\sqrt{2\pi i\left( \cos t\sinh t+\sin
t\cosh t\right) }}  \label{sp7} \\
& \quad \times \exp \left( \frac{\left( x^{2}-y^{2}\right) \sin t\sinh
t+2xy-\left( x^{2}+y^{2}\right) \cos t\cosh t}{2i\left( \cos t\sinh t+\sin
t\cosh t\right) }\right) ,  \notag
\end{align}%
which was found in \cite{Me:Co:Su} as the special case $n=1$ of a general $n$%
-dimensional expansion of the Green function in hyperspherical harmonics. We
have showed that (\ref{sp7}) is a generalization of the propagator for the
simple harmonic oscillator; see Ref.~\cite{Me:Co:Su} for more details.

The corresponding nonlinear Schr\"{o}dinger equation \cite%
{Cor-Sot:Lop:Sua:Sus}%
\begin{equation}
i\frac{\partial \psi }{\partial t}+\cos ^{2}t\ \frac{\partial ^{2}\psi }{%
\partial x^{2}}-\sin ^{2}t\ x^{2}\psi +i\sin t\cos t\left( 2x\frac{\partial
\psi }{\partial x}+\psi \right) =h\left( t\right) \left\vert \psi
\right\vert ^{2\nu }\psi  \label{sp7a}
\end{equation}%
can be rewritten in an integral form. We leave the detail to the reader.

\noindent \textbf{Acknowledgments.\/} The authors are grateful to Professor
Carlos Castillo-Ch\'{a}vez for support and reference \cite{Bet:Cin:Kai:Cas}.
We thank Professors Hank Kuiper, Alex Mahalov, and Svetlana Roudenko for
valuable comments.

\end{document}